\documentclass[12pt]{article}
\usepackage[top=0.9in,bottom=1.1in,left=1.05in,right=1.05in]{geometry}
\usepackage{amsmath,amssymb}
\usepackage{amsthm}
\usepackage{threeparttable}
\usepackage{latexsym}
\usepackage{mathrsfs,dsfont}
\usepackage{cancel}
\usepackage{comment}
\usepackage{float}
\usepackage{graphicx}
\usepackage{epstopdf}
\usepackage{color}
\usepackage{enumerate}
\usepackage{mathtools}
\usepackage{enumitem}
\allowdisplaybreaks

\mathtoolsset{showonlyrefs}
\DeclareGraphicsExtensions{.eps}
\usepackage{hyperref}

\numberwithin{equation}{section}

\newcommand{\MCR}{\mathcal{R}}

\newcommand{\MCM}{\mathcal{M}}
\newcommand{\MCF}{\mathcal{F}}
\newcommand{\MCO}{\mathcal{O}}
\newcommand{\MCD}{\mathcal{D}}

\newcommand{\MCL}{\mathcal{L}}

\newcommand{\EE}{\mathbb{E}}
\newcommand{\PP}{\mathbb{P}}
\newcommand{\RR}{\mathbb{R}}
\newcommand{\NN}{\mathbb{N}}

\newcommand{\Ltx}{\mathcal{L}_{t,x}}

\newcommand{\Vf}{V^{\eps,\delta}}

\newcommand{\vz}{v^{(0)}}

\newcommand{\pz}{{\pi^{(0)}}}

\newcommand{\wth}{w^{(3)}}

\newcommand{\Vfull}{V^{\eps,\delta}}

\newcommand{\vzo}{v^{(0,1)}}
\newcommand{\voz}{v^{(1,0)}}
\newcommand{\vztw}{v^{(0,2)}}
\newcommand{\voo}{v^{(1,1)}}
\newcommand{\vtz}{v^{(2,0)}}
\newcommand{\vthz}{v^{(3,0)}}
\newcommand{\vto}{v^{(2,1)}}

\newcommand{\wtz}{w^{(2,0)}}
\newcommand{\wthz}{w^{(3,0)}}
\newcommand{\wto}{w^{(2,1)}}

\newcommand{\aves}{\overline\lambda}

\newcommand{\eps}{\varepsilon}

\newcommand{\abs}[1]{\left|#1\right|}
\newcommand{\average}[1]{\left\langle#1\right\rangle}

\newcommand{\ud}{\,\mathrm{d}}

\newcommand{\half}{\frac{1}{2}}

\newtheorem{theo}{Theorem}[section]
\newtheorem{lem}[theo]{Lemma}

\newtheorem{rem}[theo]{Remark}
\newtheorem{prop}[theo]{Proposition}
\newtheorem{assump}[theo]{Assumption}

\newtheorem{defi}[theo]{Definition}

\begin{document}

\title{
Sub- and Super-solution Approach to Accuracy Analysis of Portfolio Optimization Asymptotics in Multiscale Stochastic Factor Markets}
\author{Jean-Pierre Fouque\thanks{Department of Statistics \& Applied Probability,
 University of California,
        Santa Barbara, CA 93106-3110, {\em fouque@pstat.ucsb.edu}. }
        \and Ruimeng Hu\thanks{Department of Mathematics and Department of Statistics \& Applied Probability, University of California,
        	Santa Barbara, CA 93106-3080,  {\em rhu@ucsb.edu}. }
        \and Ronnie Sircar\thanks{ORFE Department, Princeton University, Sherrerd Hall, Princeton NJ 08544; {\em sircar@princeton.edu.}  }
        }  
\date{\today}
\maketitle

\begin{abstract}
	
The problem of portfolio optimization when stochastic factors drive returns and volatilities has been studied in previous works by the authors. In particular, they proposed asymptotic approximations for value functions and optimal strategies in the regime where these factors are running on both slow and fast timescales. However, the rigorous justification of the accuracy of these approximations has been limited to power utilities and a single factor. In this paper, we provide an accurate analysis for cases with general utility functions and two timescale factors by constructing sub- and super-solutions to the fully nonlinear problem so that their difference is at the desired level of accuracy. This approach will be valuable in various related stochastic control problems.

\end{abstract}




\section{Introduction}

The study of the portfolio optimization problem in continuous time has a long history dating back to \cite{Me:69,Me:71}. Specifically, an agent aims to maximize her expected utility of the terminal wealth when the investment is divided between risky assets and a riskless asset, with particular types of utility functions.  In Merton's seminal work, the risky assets are assumed to follow the Black-Scholes model, where the expected returns and volatilities are constant. Since then, the problem has been studied extensively in various settings and levels of generality, including considering transaction costs \cite{DaNo:90,GuMu:13}, stochastic volatility models for risky assets (see, for instance \cite{FoSiZa:13}, and references therein), and various trading constraints (see \cite{rogers2013optimal} for a survey).

In this paper, we consider the portfolio optimization problem in a stochastic environment, where both the expected return $\mu$ and volatility $\sigma$ of a stock price $S$ are driven by two factors $Y_t$ and $Z_t$ evidenced in empirical study \cite{FoPaSiSo:11}:
\begin{equation}
\frac{\ud S_t}{S_t} = \mu(Y_t, Z_t) \ud t + \sigma(Y_t, Z_t) \ud W_t.
\end{equation}
Here $W$ is a standard Brownian motion.
The two factors are characterized by the small parameters $\eps$ and $\delta$, the fast timescale being represented by $\eps$, and the slow timescale by $1/\delta$. Along this direction, results have been developed in \cite{FoSiZa:13} about asymptotic expansions of the optimal trading strategy and the maximized utility when the utility function is general, as both $\eps$ and $\delta$ tend to zero.  

However, the accuracy of approximation was rigorously justified only in the cases with power utilities and one stochastic factor, when a distortion transformation of Hopf-Cole type introduced in \cite{zariph2001} enables a reduction to a {\em linear} PDE problem. This is only possible in the case of multiple factors in special cases; see \cite{avshsi}.
After that,  several attempts have been made to partially justify these expansions in more general cases, {\it e.g.},  in \cite{fouque2017asymptotic,hu2018asymptotic,FoHu:20}; see also \cite{huthesis} for a comprehensive review. 
The technique of sub- and super-solutions to prove the accuracy of asymptotic approximations is used in \cite{bichuch2019optimal2} for a model of power utility maximization with two nearly correlated factors corresponding to a regular perturbation problem, and in \cite{bichuch2019optimal} for optimal investment under stochastic volatility and transaction costs. 
The contribution of this paper is to rigorously justify the heuristic expansion provided in \cite{FoSiZa:13}. The methodology presented here can be adapted to the derivation of accuracy results in various contexts, as in \cite{fouque2021optimal} for instance, where fast mean-reversion is shown in data and the corresponding control problem appears as a Ricatti equation with fast mean-reverting random coefficients. Proof of accuracy of an approximate solution to this singular perturbation problem was obtained by constructing sub- and super-solutions.

We shall construct two functions close to the value function 
of the problem. These two functions act as lower and upper bounds, namely sub- and super- solutions of this problem. By requiring certain properties for each function and, most importantly, by constructing their difference to reach the desired order, one can show that the asymptotic approximations derived in \cite{FoSiZa:13} are rigorous under general utility. The requirement for the sub-solution is rather relaxed, and we shall work with a particular zeroth-order strategy. Since the super-solution acts as an upper bound of the problem, we need to ensure the property holds for all admissible strategies. 

The rest of this paper is organized as follows. In Section~\ref{sec:mainidea}, we first introduce the Merton problem in a multiscale stochastic environment and the associated Hamilton-Jacobi-Bellman (HJB) equation, and we describe the main result of the paper, Theorem~\ref{thm:main}. In Section~\ref{sec:approach}, we briefly explain our approach which consists in constructing sub- and super-solutions so that their difference is at the desired level. The asymptotic approximations derived in \cite{FoSiZa:13} are reviewed in Section~\ref{sec:existing}. Section~\ref{sec:assump} summarizes standing assumptions in this paper and some preliminary estimates which facilitate the asymptotic analysis in Section~\ref{sec:derivation}. Section~\ref{sec:derivation} is dedicated to the construction of the sub- and super-solutions, thus completing the proof of Theorem~\ref{thm:main}. We make concluding remarks in Section~\ref{sec:conclusion}.

\section{Merton Problem under Multiscale Stochastic Environment}\label{sec:mainidea}

We consider the utility maximization problem on the finite horizon $[0,T]$ with general terminal utility $U$, where the (single) underlying asset $S_t$ is driven by two factors: one fast mean-reverting $Y_t$, and one slowly varying $Z_t$:
\begin{align}\label{eq_St}
&\ud S_t = \mu(Y_t, Z_t)S_t\ud t + \sigma(Y_t, Z_t)S_t\ud W_t, \\
&\ud Y_t = \frac{1}{\eps}b(Y_t) \ud t + \frac{1}{\sqrt{\eps}}a(Y_t) \ud W_t^Y, \label{eq_Yt} \\
&\ud Z_t = \delta c(Z_t) \ud t + \sqrt\delta g(Z_t)\ud W_t^Z \label{eq_Zt},
\end{align}
Here $W, W^Y, W^Z$ are standard Brownian motions on a filtered probability space $(\Omega, \MCF, \PP)$ that are correlated:
\begin{equation}
\ud \average{W, W^Y}_t = \rho_1 \ud t, \quad \ud \average{W, W^Z}_t = \rho_2 \ud t, \quad \ud \average{W^Y, W^Z}_t = \rho_{12} \ud t,
\end{equation}
where $\abs{\rho_1} < 1$, $\abs{\rho_2} < 1$, $\abs{\rho_{12}}<1$ and $1 + 2\rho_1\rho_2\rho_{12} - \rho_1^2 - \rho_2^2 - \rho_{12}^2 >0$ to ensure the positive definiteness of the covariance matrix of $(W, W^Y, W^Z)$.  The time scales of $Y_t$ and $Z_t$ are described by the small positive parameters $\eps$ and $\delta$ respectively. We shall assume that the process $Y_t $ is ergodic and has a unique invariant distribution $\Phi$ which is independent of $\eps$ as its infinitesimal generator is of the form $\eps^{-1}{\cal L}_y$.
Further assumptions on the model parameters $\mu, \sigma, b, a, c, g$ will be list in Section~\ref{sec:assump}. In particular, they will imply that the system \eqref{eq_St}-\eqref{eq_Zt} has a unique strong solution.

Let $X_t^\pi$ be a wealth process with $\pi_t$ being the dollar amount invested in the underlying asset at time $t$ and the remaining held in a money market earning interest at rate $r$. Under the self-financing assumption, and without loss of generality, assuming $r=0$, the wealth process follows:
\begin{equation}\label{eq:wealth}
\ud X_t^\pi = \pi_t[\mu(Y_t, Z_t) \ud t + \sigma(Y_t, Z_t) \ud W_t].
\end{equation}
We are interested in the utility maximization of the terminal wealth:
\begin{equation*}
\sup_\pi \EE[U(X_T^\pi)], 
\end{equation*}
for all admissible strategies $\pi$ (see Definition~\ref{def:admissible}), and where $U(\cdot)$ is a general utility function satisfying Assumption~\ref{assump_U}. 

Restricting the problem to Markovian strategies of the form $\pi(t,x,y,z)$, we denote by $\Vfull(t, x, y, z)$ the value function 
\begin{equation}\label{def:Vfull}
\Vfull(t,x,y,z) := \sup_\pi \EE_t[U(X_T^\pi)], \quad \EE_t[\cdot] = \EE[\cdot | X_t^\pi = x, Y_t = y, Z_t = z].
\end{equation}
As described in \cite{FoSiZa:13},  the Hamilton-Jacobi-Bellman (HJB) PDE associated to this problem is,
\begin{equation}\label{eq:Vfull}
\sup_\pi Q^\pi[\Vfull](t, x, y, z) = 0, \quad \Vfull(T, x, y, z) = U(x),
\end{equation}
where $Q^\pi =\partial_t+ {\cal L}^\pi$ and ${\cal L}^\pi$ is the infinitesimal generator of $(X^\pi_t, Y_t, Z_t)$ for any Markovian strategy $\pi(t, x, y, z)$. Specifically,
\begin{align}\label{def_Q}
Q^{\pi}&=\partial_t+ \frac{1}{\eps} \MCL_y + \sqrt{\frac{\delta}{\eps}}\MCM_{yz} +  \delta\MCL_z
+ \half\pi^2\sigma(y, z)^2\partial_{xx} \\
&+ \pi\left[\mu(y,z) \partial_x +\frac{1}{\sqrt{\eps}} \rho_1a(y)\sigma(y,z) \partial_{xy} +  \sqrt{\delta}\rho_2g(z)\sigma(y, z)\partial_{xz}\right],
\end{align}
where 
\begin{equation}
\MCL_y = b(y) \partial_y + \half a^2(y) \partial_{yy}, \quad \MCL_z = c(z) \partial_z + \half g^2(z) \partial_{zz}, \quad \MCM_{yz} = \rho_{12}a(y)g(z)\partial_{yz}.
\end{equation}
We stress the fact that, in general, it is not known if the PDE \eqref{eq:Vfull} 
admits a classical solution.

The formal asymptotic expansion in the regime where $(\eps, \delta)$ are both small, performed in  \cite[Section~4]{FoSiZa:13}, 
shows 
that $ \Vfull \approx \vz + \sqrt{\eps}\voz + \sqrt{\delta}\vzo$ where $\vz$, $\voz$ and $\vzo$ are functions of $(t,x,z)$ given here in Section~\ref{sec:existing}.  
The main result of this paper is to  justify this asymptotic expansion, \emph{i.e.}, to prove the following theorem.

\begin{theo}\label{thm:main}
Under Assumptions~\ref{assump_U}--\ref{assump_paper}, the following accuracy estimate for $\Vfull$ defined in \eqref{def:Vfull} holds:
\begin{equation}\label{eq:mainexpansion}
\abs{\Vfull - \left(\vz + \sqrt{\eps}\voz + \sqrt{\delta}\vzo\right)}(t, x, y, z) \leq  \MCO(\eps + \delta),
\end{equation}
for fixed $(t, x, y, z) \in [0,T] \times \RR^+ \times \RR^2$ and sufficiently small $(\eps, \delta)$, where $\vz$, $\voz$ and $\vzo$ will be given in Section~\ref{sec:existing}. 

Additionally, the strategy $\pz$ given by \eqref{def:pz} is asymptotically optimal in the sense: 
\begin{equation}\label{eq:pi0accuracy}
\abs{\Vfull (t,x,y,z)-\EE_t[U(X_T^\pz)] } \leq  \MCO(\eps + \delta).
\end{equation}

\end{theo}

\subsection{Methodology}\label{sec:approach}

The method for proving Theorem~\ref{thm:main} is to construct two functions $V^\pm(t,x,y,z)$ as sub- and super-solutions, whose asymptotic expanded terms of the orders up to $\MCO(\sqrt\eps + \sqrt{\delta})$ coincide with $\vz + \sqrt{\eps}\voz + \sqrt{\delta} \vzo$. Specifically, we aim to find  a function $V^-$ for the sub-solution such that, for all $(x, y, z)$ and sufficiently small $(\eps, \delta)$, the following requirements are satisfied:
\begin{enumerate}[ label=(\textbf{R\arabic*})]
	\item\label{r1} The function value $V^-(T,x,y,z)$ is dominated by $U(x)$;
	\item\label{r2} The process $V^-(t, X_t^\pz, Y_t, Z_t)$ along a zero-order strategy, denoted by $\pz$, is a submartingale.
\end{enumerate} 

Thus, by the definition of $\Vfull(t, x, y, z)$, \ref{r1} and \ref{r2}, one can obtain:
\begin{equation}
\Vfull(t,x,y,z)  \geq \EE_t[U(X_T^\pz)] \geq  \EE_t[V^-(T, X_T^\pz, Y_T, Z_T)] 
 \geq V^-(t,x,y,z). \label{eq:subineq}
\end{equation}
Then, we aim to find  a function $V^+$ for the super-solution such that, for all $(x, y, z)$ and sufficiently small $(\eps, \delta)$, the following requirements are satisfied:
\begin{enumerate}[ label=(\textbf{R\arabic*}), resume]
	\item\label{r3} The function value $V^+(T,x,y,z)$ dominates $U(x)$;
	\item\label{r4} $\widehat Q[V^+](t,x,y,z) := \sup_\pi Q^\pi[V^+](t,x,y,z)$ exists and is non-positive;
	\item\label{r5} The It\^{o} integrals $\int_0^t V^+_x \pi\sigma(Y_s, Z_s) \ud W_s$, $\int_0^t V^+_y  \frac{1}{\sqrt{\eps}}a(Y_s) \ud W_s^Y$ and  $\int_0^t V^+_z \sqrt{\delta} g(Z_s) \ud W_s^Z$ are true martingales, for any admissible $\pi$.
\end{enumerate}

Then, as in  the argument used for \eqref{eq:subineq}, one can deduce:
\begin{align}
\EE_t[U(X_T^\pi)] &\leq \EE_t[V^+(T, X_T^\pi, Y_T, Z_T)] \\
&= V^+(t,x,y,z) + \EE_t\left[\int_t^T Q^\pi[V^+](s, X_s^\pi, Y_s, Z_s) \ud s \right]  \\
& \quad + \EE_t\left[\int_t^T V^+_x \pi\sigma(Y_s, Z_s) \ud W_s  \right] +\EE_t\left[\int_t^T V^+_y  \frac{1}{\sqrt{\eps}}a(Y_s) \ud W_s^Y  \right]  \\
& \quad +\EE_t\left[\int_t^T V^+_z \sqrt{\delta} g(Z_s) \ud W_s^Z  \right]  \\
& \leq V^+(t,x,y,z) + \EE_t\left[\int_t^T \widehat Q[V^+](s, X_s^\pi, Y_s, Z_s) \ud s\right]  \leq V^+(t,x,y,z) \label{eq:superineq}.
\end{align}
Taking the supremum over all admissible $\pi$ on both sides of \eqref{eq:superineq} gives $\Vfull(t, x, y ,z) \leq V^+(t, x, y, z)$.
Combining \eqref{eq:subineq} and \eqref{eq:superineq} gives
\begin{align}
 \vz + \sqrt\eps \voz + \sqrt{\delta} \vzo + o(\sqrt{\eps} + \sqrt{\delta})  &= V^-(t,x,y,z) \leq \Vfull(t, x, y, z) \leq  V^+(t, x, y, z)  \\
 &= \vz + \sqrt\eps \voz + \sqrt{\delta} \vzo + o(\sqrt{\eps} + \sqrt{\delta}).
\end{align}
In fact, the next order terms after $\vz + \sqrt{\eps}\voz  + \sqrt{\delta}\vzo$ in the construction of $V^\pm$ are $\MCO(\eps + \delta)$. Therefore $o(\sqrt{\eps} + \sqrt{\delta})$ can be replaced by  $\MCO(\eps + \delta)$ and  \eqref{eq:mainexpansion} follows.

Our choice of $V^\pm$ takes the following form:
\begin{align}\label{def:Vpm}
V^\pm &= \vz + \sqrt\eps \voz + \sqrt{\delta}\vzo + \eps \wtz + \eps^{3/2} \wthz + \eps\sqrt{\delta}\wto \\
& \quad \pm (2T-t)(\eps N_A + \delta N_B + \sqrt{\eps\delta}N_C) \pm \eps^2 F \pm \eps^{3/2}\sqrt{\delta}H \pm \eps\delta G ,
\end{align}
where $(N_A, N_B, N_C)$ are functions of $(t,x,z)$, and $(F,G,H)$ are functions of $(t,x,y,z)$. The intuition for such form is the following: (a) functions $\wtz, \wthz, \wto$ are added to eliminate terms of $\MCO(1), \MCO(\sqrt{\eps})$ and $\MCO(\sqrt{\delta})$ when applying the operator $Q^\pi$ to $V^\pm$; (b) $(N_A, N_B, N_C)$ and $(F, G, H)$ helps to fulfill \ref{r2} and \ref{r4}; (c) the coefficient $2T-t$ is for \ref{r1} and \ref{r3}. In Section~\ref{sec:derivation}, we shall show how these functions are determined and why they can be chosen as functions of particular variables such that the requirements \ref{r1}--\ref{r5} are satisfied. In the rest of this section, we briefly review the existing derivations of $\vz, \voz, \vzo$, the definition of $\wtz, \wthz, \wto$, the standing assumptions in this paper, and some preliminary estimates.

\subsection{Multiscale  asymptotic expansions}\label{sec:existing}
Generally, closed-form solutions are barely available for HJB equations. In our setup, we do not even know if $\Vfull$ solves \eqref{eq:Vfull} in the viscosity sense. 
In \cite{FoSiZa:13}, a first-order expansion of $\Vf$ around small $(\eps, \delta)$ is formally derived via singular and regular perturbation techniques. Since the formulas of these terms and the equations they satisfy play an important role in proving our main theorem, we summarize them below for readers' convenience.  For detailed derivations, we refer the readers to \cite[Section~4]{FoSiZa:13}  and \cite[Sections~2 and 3]{FoHu:20}.

The combined expansion in slow and fast scales of $\Vf$ is of the following form:
\begin{equation}\label{p2_eq_Vfexpansion}
\Vf = \vz + \sqrt{\eps}\voz + \sqrt{\delta}\vzo + \eps\vtz + \delta \vztw + \sqrt{\eps\delta}\voo + \cdots,
\end{equation}
where the superscript of $v$ corresponds to the powers in $\sqrt\eps$ and $\sqrt\delta$ and where  $v^{(0,0)}$ has been rewritten as $v^{(0)}$.  To precisely give the equations which identify these terms, we introduce the following notations, following \cite{FoSiZa:13}. Denote by $\average{\cdot}$ the average with respect to the $\eps$-independent invariant distribution $\Phi$ of $Y$: $\average{g} = \int g(y) \, \Phi(\mathrm{d}y)$, and by $M(t, x; \lambda)$ the solution to the classical Merton PDE where $\mu$ and $\sigma$ are constants:
\begin{equation}\label{eq:classicalMerton}
M_t - \half \lambda^2 \frac{M_x^2}{M_{xx}} = 0, \quad M(T, x; \lambda) = U(x), \quad \lambda = \mu/\sigma \quad \mbox{(Sharpe ratio)}.
\end{equation}
We define the associate risk-tolerance function $$R(t, x; \lambda) := - \frac{M_x(t, x; \lambda)}{M_{xx}(t, x; \lambda)},$$ and the differential operators:
\begin{equation}
D_k(\lambda) = R(t, x; \lambda)^k \partial_x^k, \quad k = 1, 2, \cdots; \quad \Ltx(\lambda) = \partial_t + \half \lambda^2 D_2(\lambda) + \lambda^2 D_1(\lambda).
\end{equation}
We denote the ``square-averaged'' Sharpe ratio $\aves(z)=\sqrt{\average{\lambda^2(\cdot,z)}}$, and the version of $D_k(\lambda)$ that will be used in the sequel is $D_k(\aves) = R(t,x;\aves(z))^k \partial_x^k$. We shall use $D_k$ for brevity (omitting the argument $\aves$).
We also define the averaged Sharpe ratio: 
$\widehat \lambda(z) = \average{\lambda(\cdot,z)}$.
Now we are ready to present the formulations of $\vz$, $\voz$ and $\vzo$.

\begin{prop} (\cite[Section~4]{FoSiZa:13}  and \cite[Sections~2 \& 3]{FoHu:20})
\begin{enumerate}[label = (\roman*)]	
	\item The \emph{leading order term} $\vz$ is defined as the classical solution to the Merton PDE 
	\begin{equation}\label{eq:vz}
	\vz_t - \half\aves^2(z)\frac{\left(\vz_x\right)^2}{\vz_{xx}} = 0, \quad \vz(T,x,z) = U(x).
	\end{equation}
	Since it  possesses a unique solution (see \cite[Proposition~2.2]{FoHu:20}), we have
	\begin{equation}\label{eq:vzvsmerton}
	\vz(t,x,z) = M(t,x;\aves(z)).
	\end{equation}	
	\item The \emph{first order correction} in the fast variable $\voz$ is defined as the solution to the linear PDE:
	\begin{equation}\label{eq:voz}
	\Ltx(\overline \lambda(z)) \voz 
	= \half \rho_1B(z)D_1^2\vz, \qquad \voz(T,x,z) = 0,
	\end{equation}
	where $B(z) = \average{\lambda(\cdot,z)a(\cdot)\partial_y\theta(\cdot,z)}$, and $\theta$ is a solution of the Poisson equation 
	\begin{equation}\label{def:theta}
	\MCL_y\theta(y,z) = \lambda^2(y,z) - \aves^2(z). 
	\end{equation}
		It is explicitly given in terms of $\vz$ by
	\begin{equation}\label{def:voz}
	\voz(t,x,z) = -\half(T-t)\rho_1B(z)D_1^2\vz(t,x,z).
	\end{equation}

	\item The \emph{first order correction} in the slow variable $\vzo$ is defined as the solution to the linear PDE:
	\begin{equation}\label{eq:vzo}
	\Ltx(\overline{\lambda}(z))\vzo 
	= \rho_2\widehat \lambda(z)g(z)\frac{\vz_x}{\vz_{xx}}\vz_{xz} , 	\qquad\vzo(T,x,z) = 0,
	\end{equation}
and $\vzo$ can be expressed in terms of $\vz$ by
	\begin{align}
	\vzo(t,x,z) &= \half(T-t)\rho_2\widehat{\lambda}(z)g(z)D_1\vz_z(t,x,z)\\ 
	&=\half(T-t)^2\rho_2\widehat{\lambda}(z)\aves(z)\aves'(z)g(z)D_1^2\vz(t,x,z).\label{def:vzo}
	\end{align}
	
	\item The $z$-derivatives of the leading order term $\vz$ and the risk-tolerance function $R$ satisfy
	\begin{align}\label{eq:vegagamma}
	\vz_z(t,x,z) &= (T-t)\aves(z)\aves'(z)D_1\vz(t,x,z), \\ 
	R_z(t, x; \overline\lambda(z)) &= (T-t) \overline \lambda(z)\overline \lambda'(z) R^2 R_{xx}(t, x; \overline{\lambda}(z)).
	\end{align}
		
	\item The term $\vtz$ solves the linear PDE: $\MCL_y\vtz + \Ltx(\lambda(y,z))\vz=0$, 
	 and so has the form
	\begin{equation}
	\vtz(t,x,y,z) = -\half \theta(y,z) D_1\vz + C_1(t, x, z),
	\end{equation}
	where $\theta$ solves \eqref{def:theta}.
	\item The term $\vthz$ solves the linear PDE:  $$\MCL_y \vthz + \Ltx(\lambda(y,z))\voz = \half \rho_1\lambda(y, z)a(y)\partial_y\theta(y)D_1^2\vz$$ and so has the form
	\begin{multline}
	\vthz(t,x,y,z) = \half (T-t) \rho_1\theta(y,z) B(z) (\half D_2 + D_1)D_1^2 \vz + \half \rho_1 \theta_1(y,z) D_1^2 \vz \\
	+ C_2(t, x, z),
	\end{multline}
	where $\theta_1(y, z)$ solves the Poisson equation
	$\MCL_y \theta_1(y,z) = \lambda(y,z)a(y)\partial_y\theta(y,z) - B(z).$
	
	\item The term $\vto$ solves the linear PDE: $$\MCL_y \vto + \Ltx(\lambda(y,z)) \vzo + \rho_2 g(z)\lambda(y,z)D_1\vz_z = 0. $$ With \eqref{eq:vzvsmerton} and \eqref{def:vzo}, $\vto$ is given by
	\begin{align}
	\vto(t,x,y,z) &= -\half (T-t)^2\theta(y,z) \rho_2 \widehat \lambda(z) \overline \lambda(z) \overline{\lambda}'(z)g(z)(\half D_2 + D_1)D_1^2 \vz \\
	&\quad - \rho_2(T-t)\theta_2(y,z)g(z) \overline{\lambda}(z)\overline{\lambda}'(z) D_1^2\vz + C_3(t, x, z), 
	\end{align}
	where $\theta_2(y,z)$ solves the Poisson equation
$	\MCL_y \theta_2(y, z) = \lambda(y, z) - \widehat \lambda(z). $
\end{enumerate}
\end{prop}

In the sequel, when deriving the concrete formula of $V^\pm$, we shall choose $\wtz$, $\wthz$ and $\wto$ to be the corresponding terms in the expansion of $\Vfull$ with $C_i(t,x,z) \equiv 0$, $i = 1, 2, 3$. That is, we choose
\begin{align}
\wtz(t,x,y,z) &= -\half \theta(y,z) D_1\vz, \label{def:wtz}\\
\wthz(t,x,y,z) &= \half (T-t) \rho_1\theta(y,z) B(z) (\frac{D_2}{2} + D_1)D_1^2 \vz + \half \rho_1 \theta_1(y,z) D_1^2 \vz, \label{def:wthz}\\
\wto(t,x,y,z) &= -\half (T-t)^2\theta(y,z) \rho_2 \widehat \lambda(z) \overline \lambda(z) \overline{\lambda}'(z)g(z)(\frac{D_2}{2} + D_1)D_1^2 \vz \\
& \quad - \rho_2(T-t)\theta_2(y,z)g(z) \overline{\lambda}(z)\overline{\lambda}'(z) D_1^2\vz.  \label{def:wto}
\end{align}

\subsection{Model assumptions and preliminary estimates}\label{sec:assump}

We first make precise the regularity assumptions on the utility function $U$, on the risk tolerance $-U'/U''$, and on the inverse marginal utility $(U')^{(-1)}$. They will be satisfied by mixtures of power utilities or sums of inverse marginal power utilities for instance, and we refer to Appendix A in 
\cite{fouque2017asymptotic} for further details. The advantage of these mixtures is that the Arrow-Pratt risk aversion ($-xU''/U'$) is wealth dependent as opposed to constant for pure power utilities.

\begin{assump}\label{assump_U}
	We make the following assumptions on the utility function $U(x)$:
	\begin{enumerate}
		\item\label{p1_assump_Uregularity}  U(x) is $C^9(0,\infty)$, strictly increasing, and strictly concave and satisfies the following conditions (Inada and asymptotic elasticity):
		\begin{equation}\label{p1_eq_usualconditions}
		U'(0+) = \infty, \quad U'(\infty) = 0, \quad \text{AE}[U] := \lim_{x\rightarrow \infty} x\frac{U'(x)}{U(x)} <1.
		\end{equation}
		\item\label{p1_assump_Urisktolerance} Assume the risk tolerance $R(x) := -U'(x) / U''(x)$ satisfies $R(0) = 0$, strictly increasing, $R'(x) < \infty$ on $[0,\infty)$, and there exists $K\in\RR^+$, such that for $x \geq 0$, and $ 2\leq i \leq 7$,
		\begin{equation}\label{p1_assump_Uiii}
		\abs{\partial_x^{(i)}R^i(x)} \leq K.
		\end{equation}
		\item\label{p1_assump_Ugrowth} Define the inverse function of the marginal utility $U'(x)$ as $I: \RR^+ \to \RR^+$, $I(y) = U'^{(-1)}(y)$, and assume that, for some positive $\alpha$, $I(y)$ satisfies the polynomial growth condition:
		\begin{equation}\label{p1_cond_I}
		I(y) \leq \alpha + \kappa y^{-\alpha}.
		\end{equation}
	\end{enumerate}
\end{assump}

Now, we make precise the definition of a Markovian admissible strategy.

\begin{defi}[Admissibility]\label{def:admissible}
	A Markovian strategy $\pi$ is admissible if $X^\pi_t$ stays positive a.s. for all $t\in [0,T]$ and
	\begin{align}
	&\EE\int_0^T \left( \pi(t, X_t^\pi, Y_t, Z_t)\sigma(Y_t, Z_t) \vz_x(t, X_t^\pi, Z_t)\right) ^2 \ud t < \infty,
	\label{admissibility1}\\
	&\EE\int_0^T \left( D_1 \vz(t, X_t^\pi, Z_t)\right) ^2 \ud t  < \infty.\label{admissibility2}
	\end{align}
\end{defi}

Next, we make the following technical assumptions on the model parameters and the various quantities appearing in our expansion. In particular, it involves several Poisson equations for which we assume that the solutions are bounded.

\begin{assump}\label{assump_paper} 
	
	\begin{enumerate}
		\item\label{p2_assump_valuefuncSZadd} For any starting points $(s, y, z)$ and fixed $(\eps, \delta)$, the system of SDEs \eqref{eq_St}--\eqref{eq_Yt}--\eqref{eq_Zt} has a unique strong solution $(S_t, Y_t, Z_t)$. The  process $Y$ is ergodic and has a unique invariant distribution $\Phi$ (independent of $\eps$).

		\item The following functions are bounded with bounded derivatives:
		\begin{equation}
		\lambda(y,z), g(z), c(z), a(y), B(z), \widehat{\lambda}(z),  \overline \lambda(z), 
		B_1(z), \theta(y,z), \theta_i(y,z), 1 \leq i \leq 11, 
		\end{equation}
		where $\theta_i$ are solutions of the Poisson equations:
		\begin{align}
		&\MCL_y \theta(y,z) = \lambda^2(y,z) - \overline \lambda^2(z), \; \overline \lambda^2(z) = \average{\lambda^2(\cdot, z)}, \\
		&\MCL_y \theta_1(y,z) = \lambda(y,z)a(y)\partial_y \theta(y,z) - B(z), \;  B(z) = \average{\lambda(\cdot, z)a(\cdot)\partial_y \theta(\cdot, z)},\\
		& \MCL_y \theta_2(y,z) = \lambda(y,z) -\widehat \lambda(z), \; \widehat \lambda(z) = \average{\lambda(\cdot, z)}, \\	
		&\MCL_y \theta_3(y,z) = a(y)\lambda(y,z)\partial_y \theta_1(y,z) - B_1(z), \; B_1(z) = \average{a(\cdot)\lambda(\cdot, z) \partial_y \theta_1(\cdot, z)},\label{def:B1}\\
		&\MCL_y \theta_4(y,z) = \theta(y,z) \lambda^2(y,z) - \average{\theta\lambda^2}, \label{def:theta4}\\
		& \MCL_y \theta_5(y,z) = \theta(y,z) - \average{\theta}, \label{def:theta5}\\
		&\MCL_y \theta_6(y,z) = \partial_{yz} \theta(y,z) - \average{\partial_{yz} \theta(\cdot, z)}, \label{def:theta6}\\
		&\MCL_y \theta_7(y,z) =\partial_y \theta (y,z) - \average{\partial_y \theta(\cdot, z)}, \label{def:theta7} \\
		&\MCL_y \theta_8(y,z) = a(y)\lambda(y,z)\partial_y \theta_2(y,z) - \average{a(\cdot)\lambda(\cdot, z)\partial_y \theta_2(\cdot, z)}, \label{def:theta8} \\
		&\MCL_y \theta_9(y, z) = a^2(y) \left(\partial_y \theta\right)^2(y,z) - \average{a^2(\cdot)\left(\partial_y \theta\right)^2(\cdot, z)}, \label{def:theta9} \\
		&\MCL_y \theta_{10}(y,z) = a(y) \partial_{yz} \theta(y,z) - \average{a(\cdot)   \partial_{yz} \theta(\cdot, z)}, \label{def:theta10} \\
		&\MCL_y \theta_{11}(y,z) = a(y) \partial_y \theta(y,z) - \average{a(\cdot) \partial_{yz} \theta(\cdot, z)}. \label{def:theta11} 
		\end{align}
	Moreover, $\aves(z)$ is bounded away from 0.
		\end{enumerate}
	\end{assump}

With all the notations and assumption introduced, we obtain the following propositions by lengthy but straightforward calculations.  We omit the proofs here and refer to \cite{fouque2017asymptotic,FoHu:20}.
\begin{prop}\label{prop_estimate}
	Under the above assumptions, the functions $v^{(i,j)}$ and $w^{(i,j)}$, $i+j >0$ (cf.  \eqref{def:voz}, \eqref{def:vzo}, \eqref{def:wtz}, \eqref{def:wthz}  and \eqref{def:wto}) satisfy 
	\begin{align}
	&v^{(i,j)} \leq h(y,z) D_1\vz, \text{ with a bounded function } h(y,z),\\
	&v^{(i,j)}_x \leq h(y,z) \vz_x, \text{ with a bounded function } h(y,z), \\
	&v^{(i,j)}_{xx} \leq h(y,z) \vz_{xx}, \text{ with a bounded function } h(y,z),
	\end{align}
	where $h(y,z)$ denotes a bounded function in $y$ and $z$, and may vary from case to case. Similar inequalities hold for $w^{(i,j)}$. In particular, one has 
	\begin{equation}
	D_1^iD_2^j D_1^k \vz \leq h(z) D_1 \vz,  \forall i, j, k \in \NN^+
	\end{equation}
	with a bounded function $h(z)$.
\end{prop}

\begin{prop}[Proposition 2.6 in \cite{FoHu:20}]\label{prop_risktolerance}
	Under Assumption~\ref{assump_U} of the general utility, the risk-tolerance $R(t,x;\aves(z))$ function satisfies the following:  $\exists K_j >0$ for $0 \leq j  \leq  6$, such that $\forall (t,x,\aves(z)) \in [0,T) \times \RR^+ \times \RR$,
	\begin{equation}
	\abs{R^j(t,x;\aves(z))(\partial_x^{(j+1)}R(t,x;\aves(z)))} \leq K_j.
	\end{equation}
	Or equivalently, $\forall 1 \leq j \leq 7$, there exists $\widetilde K_j >0$, such that $\forall (t,x,z) \in [0,T) \times \RR^+ \times \RR$,
	\begin{equation}
	\abs{\partial_x^{(j)} R^j(t,x;\aves(z))} \leq \widetilde K_j.
	\end{equation}
	Moreover, $0 \leq R(t, x; \overline \lambda(z)) \leq K_0 x$, and
	 the following quantities are uniformly bounded: $RR_{xxz}$, $R^2R_{xxxz}$, $R_{xzz}$, $RR_{xxzz}$ and $R^2R_{xxxzz}$.
\end{prop}

\section{Proof of the Main Theorem}\label{sec:derivation}
Recall the sub- and super-solution we shall construct are of the form
\begin{align}\label{eq:Vpm}
V^\pm &= \vz + \sqrt\eps \voz + \sqrt{\delta}\vzo + \eps \wtz + \eps^{3/2} \wthz + \eps\sqrt{\delta}\wto \\
& \quad \pm (2T-t)(\eps N_A + \delta N_B + \sqrt{\eps\delta}N_C) \pm \eps^2 F \pm \eps^{3/2}\sqrt{\delta}H \pm \eps\delta G.
\end{align}
This section is dedicated to identify the terms $(N_A, N_B, N_C)$ which will be functions of $(t, x, z)$ and $(F, G, H)$  which will be functions of $(t, x, y, z)$ in $V^\pm$.

\subsection{Sub-solution}\label{sec:sub}

We shall first work with the process $V^-(t, X_t^\pz, Y_t, Z_t)$ along the given zeroth order strategy $\pz$:
\begin{equation}\label{def:pz}
\pz(t, x, y, z) := \frac{\lambda(y, z)}{\sigma(y, z)} R(t, x; \overline{\lambda}(z)) =  -\frac{\lambda(y, z)}{\sigma(y, z)} \frac{\vz_x(t, x, z)}{\vz_{xx}(t, x, z)}.
\end{equation}
\begin{lem}\label{lemma:pzadmissible}
The strategy $\pz$ is admissible (in the sense of Definition \ref{def:admissible}).
\end{lem}
\begin{proof}
From \eqref{eq:wealth} and \eqref{def:pz}, we have
\begin{equation}
\ud X_t^\pz = R(t,X_t^\pz;\overline\lambda(Z_t))\left[\lambda^2(Y_t,Z_t)\ud t+\lambda(Y_t,Z_t) \ud W_t\right].
\end{equation}
Using that $0\leq R(t,x;\overline\lambda(z)) \leq K_0 x$,  and the boundedness of $\lambda$, one deduces that $ X_t^\pz$ is a proper exponential and stays positive. Moreover, it has $p^{th}$-moments for any $p \in \NN$, uniformly in $t \in [0,T]$. Next we observe that the condition \eqref{admissibility1} applied to $\pz$ reduces to the condition \eqref{admissibility2} after using the boundedness of $\lambda$. Thus it suffices to verify \eqref{admissibility2} when $\pi = \pz$.

To this end, we recall the $H$-transform used in \cite{KaZa:14}. Let $H: \RR \times [0,T] \times \RR \to \RR^+$ be defined by
\begin{equation}\label{def:Hfunction}
\vz_x(t, H(x, t, \overline \lambda(z)), z) = e^{-x - \half \aves^2(z)(T-t)}.
\end{equation}
It satisfies the heat equation 
\begin{equation}\label{eq:Heat}
H_t + \half \aves^2(z) H_{xx} = 0,
\end{equation}
with the terminal condition $H(t, x, \aves(z) = I(e^{-x}))$. Now define the spatial inverse function $H^{(-1)}(y, t, \aves(z)): \RR^+ \times [0,T] \times \RR \to \RR$, i.e., $H(H^{(-1)}(y, t, \aves(z), t, \aves(z)) = y$. Using \eqref{def:Hfunction} and   $0\leq R(t,x;\overline\lambda(z)) \leq K_0 x$, it remains to show
\begin{equation}\label{eq:vzadmissible}
\EE\int_0^T (X_t^\pz)^2 e^{-2H^{(-1)}(X_t^\pz, t, \aves(Z_t)) - \aves^2(Z_t)(T-t)} \ud t < \infty. 
\end{equation}
To further proceed, we need a lower bound for the inverse function $H^{(-1)}$, or equivalently an upper bound for $H$ since $H$ is positive and strictly increasing. Using the fact that $H$ solves \eqref{eq:Heat} with a terminal condition $I(e^{-x}) \leq \alpha + \kappa e^{\alpha x}$, $\aves(z)$ is bounded above and below away from zero, and $t \in [0,T]$, one can deduce that, (by writing down the solution of $H$ as the convolution of $I(e^{-x})$ with the heat kernel)
\begin{equation}
H(x, t, \aves(z)) \leq C e^{Cx}, \text{ thus } H^{(-1)} (x, t, \aves(z))  \geq \frac{1}{C} \log(x/C), \text{ for a generic contant } C.
\end{equation}
Consequently, equation \eqref{eq:vzadmissible} is bounded by
\begin{equation} 
\EE\int_0^T (X_t^\pz)^p  \ud t, \quad p = 2 - 2/C.
\end{equation}
which is finite as we have shown that $X_t^\pz$ has bounded moments of any other. Therefore, $\pz$ is admissible.
\end{proof}

For the derivations presented below, the submartingality  requirement \ref{r2} will enable us to pin down certain formulas for $(N_A, N_B, N_C)$ and $(F, G, H)$ up to some constants $(C_A, C_B, C_C)$. Then with sufficiently large choices of $(C_A, C_B, C_C)$ and sufficiently small $(\eps, \delta)$, \ref{r1} is fulfilled.

\subsubsection{The submartingality requirement \ref{r2}}\label{sec:subr2}
To fulfill it, we first write down the operator $Q^\pi$ with $\pi = \pz$ given in \eqref{def:pz}:
\begin{equation}
Q^\pz = \Ltx(\lambda(y,z)) + \frac{1}{\eps} \MCL_y + \sqrt{\frac{\delta}{\eps}}\MCM_{yz} + \delta \MCL_z + \frac{1}{\sqrt{\eps}} \rho_1 a(y)\lambda(y,z) D_1 \partial_y + \sqrt{\delta}\rho_2 g(z)\lambda(y,z) D_1\partial_z.
\end{equation}
In the sequel, to avoid cumbersome notations, we will systematically omit the variables of all functions. Consider  $(N_A, N_B, N_C)$ as functions of $(t, x, z)$ and using the formula of $\vz, \voz, \vzo, \wtz, \wthz, \wto$, one has:
\begin{align}
&Q^\pz[V^-] \\
&= \eps\left(\Ltx(\lambda)\wtz - \MCL_y F + \rho_1 a\lambda D_1 \wthz_y - \Ltx(\lambda)[(2T-t)N_A]\right) \\
& \quad + \sqrt{\eps\delta} \left(-\Ltx(\lambda)[(2T-t)N_C] - \MCL_y H + \MCM_{yz} \wtz + \rho_1 a\lambda D_1 \wto_y + \rho_2 g\lambda D_1 \voz_z\right)\\
& \quad + \delta \left(-\Ltx(\lambda)[(2T-t)N_B] + \MCL_z \vz + \rho_2 g\lambda D_1\vzo_z - \MCL_y G\right)  + higher \; order \;  terms \\
& := \eps \text{I}_\eps + \sqrt{\eps\delta}\text{I}_{\eps\delta} + \delta \text{I}_\delta + h.o.t. \label{eq:QpzVminus}
\end{align}
Then, by It\^{o} formula,
\begin{align}
\ud V^-(t, X_t^\pz, Y_t, Z_t) = Q^\pz[V^-] \ud t +    V^-_x \pz\sigma \ud W_t + V^-_y  \frac{1}{\sqrt{\eps}}a \ud W_t^Y  +  V^-_z\sqrt{\delta} g\ud W_t^Z. 
\end{align}
Thus it suffices to show that
\begin{enumerate}[ label=(\textbf{R2-\arabic*})]
	\item\label{r2-1} $Q^\pz[V^-] \geq 0$ for all $(t, x,y,z)$ and sufficiently small $(\eps, \delta)$; 
	\item\label{r2-2} The It\^{o} integral terms are true martingales.
\end{enumerate}

For the item \ref{r2-1} we first require that $\text{I}_\eps$, $\text{I}_{\eps\delta}$ and $\text{I}_\delta$ are strictly positive. This will enable us to determine the forms of $(N_A, N_B, N_C)$ up to some constants and the formulas of $(F,G,H)$. Precisely speaking, the form of $N_A$ is determined by the necessary conditions $\average{\text{I}_\eps} > 0$, and $F$ is the solution to $\text{I}_\eps - \average{\text{I}_\eps} = 0$. $N_B, N_C, G, H$ are determined in a similar manner, and we present the detailed computations for the pair $(N_A, F)$ as follows.

Regarding $\MCO(\eps)$ terms in $Q^\pz[V^-]$,  we have $\text{I}_\eps = \Ltx(\lambda)\wtz - \MCL_y F + \rho_1 a\lambda D_1 \wthz_y - \Ltx(\lambda)[(2T-t)N_A]$, and we compute
\begin{align}
\Ltx(\lambda)\wtz &= -\half \theta(y,z) \Ltx(\lambda) D_1\vz  = -\half \theta(y,z)(\lambda^2 - \overline \lambda^2)(\half D_2 + D_1) D_1\vz, \\
\rho_1 a\lambda D_1\wthz_y &= \rho_1 a\lambda D_1 \left(\half (T-t)\rho_1 \theta_y B(z)(\half D_2 + D_1 ) D_1^2 \vz + \half \rho_1 \theta_{1y} D_1^2 \vz\right),
\end{align}
thus
\begin{align}\label{eq_Ieps}
\average{\text{I}_\eps} & = \average{\Ltx(\lambda)\wtz} + \average{\rho_1a\lambda D_1 \wthz_y}  - \Ltx(\overline{\lambda}) [(2T-t)N_A] \\
& = -\half \left(\average{\theta\lambda^2} - \average{\theta} \overline{\lambda}^2\right) (\half D_2 + D_1) D_1\vz + \half (T-t)\rho_1^2 B(z)^2 D_1 (\half D_2 + D_1) D_1^2 \vz \\
& \quad + \half \rho_1^2 B_1(z) D_1^3 \vz - \Ltx(\overline{\lambda}) [(2T-t)N_A], \label{eq:Ieps}
\end{align}
where $B_1(z)$ is define in \eqref{def:B1}. Note that all terms in \eqref{eq:Ieps} except the last one are bounded by a multiple of $D_1\vz$ by Assumption~\ref{assump_paper} and Proposition~\ref{prop_estimate}. Therefore, we can choose
\begin{equation}\label{def:NA}
N_A = C_A D_1\vz, \text{ for some constant } C_A.
\end{equation}
Then the last term in \eqref{eq:Ieps} becomes 
\begin{equation}
- \Ltx(\overline{\lambda}) [(2T-t)N_A] = -(2T-t) \Ltx(\overline{\lambda})C_AD_1\vz + C_AD_1\vz = C_AD_1 \vz
\end{equation}
 as $\Ltx(\overline{\lambda})D_1\vz = D_1\Ltx(\overline{\lambda})\vz = 0$. Thus, the choice \eqref{def:NA} indeed does the job of making $\average{\text{I}_\eps}$ positive for sufficiently large $C_A$. 

We next derive the formula for $F$. Identifying $F$ as the solution to $\text{I}_\eps - \average{\text{I}_\eps} = 0$, we have
\begin{align*}
\MCL_y F &= -\half \left(\theta\lambda^2 - \theta\overline{\lambda}^2 - \average{\theta\lambda^2} + \average{\theta}\overline{\lambda}^2\right)(\half D_2 + D_1 )D_1 \vz \\
& \quad + \half (T-t)\rho_1^2 (a\lambda \theta_y - B(z)) B(z) D_1 (\half D_2 + D_1) D_1^2 \vz + \half \rho_1^2 (a\lambda \theta_{1y} - B_1(z)) D_1^3\vz \\
& \quad - (2T -t )(\lambda^2 - \overline \lambda^2) (\half D_2 + D_1) C_A D_1 \vz,
\end{align*}
which yields a formula for $F$:
\begin{align}
F(t,x,y,z) &= - \half \left(\theta_4 - \overline\lambda^2 \theta_5\right)(\half D_2 + D_1) D_1\vz + \half (T-t) \rho_1^2\theta_1 B(z) D_1(\half D_2 + D_1) D_1^2 \vz \nonumber\\
& \quad + \half \rho_1^2 \theta_3 D_1^3 \vz - (2T-t)\theta (\half D_2 + D_1) C_A D_1\vz, \label{def:F}
\end{align}
Here $\theta_3, \theta_4$, and $\theta_5$ solve the Poisson equations \eqref{def:B1}, \eqref{def:theta4} and \eqref{def:theta5}, respectively.
With such choices of $N_A$ and $F$, we are able to let $\text{I}_\eps = \average{\text{I}_\eps} > 0$.

Regarding the $\MCO(\delta)$ terms, with the choice of $N_B = C_B D_1 \vz$, one deduces that
\begin{equation}
\average{\text{I}_\delta} = \MCL_z \vz + \rho_2 g(z)\widehat \lambda D_1 \vzo_z + N_B.
\end{equation}
By the Vega-Gamma relations \eqref{eq:vegagamma}, Assumption~\ref{assump_paper} and Proposition~\ref{prop_estimate},  it suffices to choose a large $N_B$ so that  $\average{\text{I}_\delta}$ is strictly positive. For identifying $G$ from $\text{I}_\delta - \average{\text{I}_\delta} = 0$ , we first write down
\begin{equation}
\MCL_y G = \rho_2 g (\lambda - \widehat \lambda) D_1 \vzo_z - (2T-t) (\lambda^2 - \overline{\lambda}^2) (\half D_2 + D_1) C_B D_1 \vz,
\end{equation}
and then obtain for $G$:
\begin{equation}
G(t,x,y,z) = \rho_2 g\theta_2(y,z) D_1 \vzo_z - (2T-t) \theta(y,z) (\half D_2 + D_1) C_B D_1 \vz. \label{def:G}
\end{equation}

For terms of  $\MCO(\sqrt{\eps\delta})$, a similar derivation yields
\begin{equation}
N_C = C_C D_1\vz, \text{ for sufficiently large } C_C,
\end{equation}
and
\begin{align}
H(t,x,y,z) &= \rho_2 g \theta_2 D_1 \voz_z - \half \theta_6 D_1 \vz - \half \theta_7 \overline{\lambda}\overline{\lambda}' (T-t) RR_{xx}D_1\vz \\
& \quad -\half \theta_7 (T-t) \overline \lambda \overline \lambda' (R_x - 1) D_1 \vz - \rho_1 \theta_1 \half (T-t)^2 \rho_2 \widehat\lambda \overline{\lambda}\overline{\lambda}' g D_1 (\half D_2  + D_1) D_1^2 \vz \nonumber \\
& \quad - \rho_1 \rho_2 \theta_8 g(T-t) \overline{\lambda}\overline{\lambda}' D_1^3 \vz - (2T-t) \theta (\half D_2 + D_1) C_c D_1 \vz, \label{def:H}
\end{align}
where $\theta_6, \theta_7$ and  $\theta_8$ solve the Poisson equations \eqref{def:theta6}, \eqref{def:theta7} and \eqref{def:theta8}, respectively.

The next step is to ensure that $Q^\pz[V^-] \geq 0$, i.e. the higher order terms in \eqref{eq:QpzVminus} are indeed negligible and can be dominated by terms of $\MCO(\eps + \delta)$. By straightforward but tedious calculation, with Assumption~\ref{assump_paper}, and Propositions~\ref{prop_estimate} and \ref{prop_risktolerance} we can verify that all terms higher than $\MCO(\eps + \delta)$ can be bounded $f(y,z, C_i) D_1\vz$, where the function $f(y,z, C_i)$ is bounded in $(y, z)$ and linear in $C_i$, for $i = A, B, C$. This is because $(F, G, H)$ contribute to higher order terms and their formulas contain linear functions in $(C_A, C_B, C_C)$. On the other hand, we have  $\text{I}_\eps = \average{\text{I}_\eps} \geq \tilde f(y,z) D_1\vz + C_A D_1\vz$, $\text{I}_\delta = \average{\text{I}_\delta} \geq \tilde f(y,z) D_1\vz + C_B D_1\vz$ and $\text{I}_{\eps\delta} = \average{\text{I}_{\eps\delta}} \geq \tilde f(y,z) D_1\vz + C_C D_1\vz$, which are the coefficients at order $\eps, \delta$ and $\sqrt{\eps\delta}$. The function $\tilde f(y, z)$ is also bounded in $(y, z)$ and may vary from case to case, but free of $(C_A, C_B, C_C)$. More precisely, one has the following
\begin{multline}
Q^\pz[V^-] \geq \eps (\tilde f(y,z) + C_A D_1\vz) + \delta (\tilde f(y,z) + C_B D_1\vz) + \sqrt{\eps\delta}(\tilde f(y,z) + C_C D_1\vz) \\
+  \sum_{i+j > 1} \eps^i \delta^j f(y,z, C_i) D_1\vz.
\end{multline}
Therefore, one can first choose $\eps < \eps'$ and $\delta < \delta'$ such that the coefficients of $C_i$ are positive, then for $C_i > C_i'$, $i = A, B, C$,  $Q^\pz[V^-]$ is always non-negative.

We now take care of  \ref{r2-2} . With our choice of $N_i = C_i D_1\vz$ for $i = A, B, C$ and choice of $(F, G, H)$ (\emph{cf.} equations \eqref{def:F}, \eqref{def:G}, \eqref{def:H}), we observe that terms in $V^-_x \pz \sigma = \lambda D_1 V^-$ are all of the form $h(y,z) \MCD\vz$, with $\MCD$ being the following operators:
\begin{equation*}
D_1, D_1^3, D_1^2, D_1D_2D_1^2, D_1^4, D_1D_2D_1, D_1^2D_2D_1^2, D_1^5, D_1RR_{xx}D_1, D_1(R_x-1)D_1, D_1^2\partial_z, D_1^2\partial_zD_1^2.
\end{equation*}
Under model assumptions, and with Propositions~\eqref{prop_estimate} and \eqref{prop_risktolerance} we deduce that $h(y,z)$ is always bounded and $\MCD\vz$ can be bounded by a multiple of $D_1\vz$. Therefore, for $\int_0^t V^-_x \pz \sigma  \ud W_t$ being a martingale, we essentially require that $\int_0^t D_1\vz \ud W_t$ is square integrable, which is fulfilled by Lemma \ref{lemma:pzadmissible} and \eqref{admissibility2}. Repeating this argument with similar calculation for $\int_0^t  V^-_y \frac{1}{\sqrt{\eps}}a \ud W_t^Y$ and $\int_0^t  V^-_z \sqrt{\delta} g \ud W_t^Z$, we claim such choices of $N_i = C_i D_1\vz$ and $(F, G, H)$ satisfy \ref{r2}.

\subsubsection{The requirement \ref{r1}: $U(x) \geq V^-(T, x, y, z)$}\label{sec:subr1}
With our choices of $N_i = C_i D_1\vz$ and $(F, G, H)$, \ref{r1} essentially requires
\begin{align*}
U(x) &\geq U(x) + \eps \wtz(T,x,y,z) + \eps^{3/2} \wthz(T,x,y,z) - T(\eps N_A + \delta N_B + \sqrt{\eps\delta} N_C)(T,x,z)\\
&\quad - \eps^2 F(T,x,y,z) - \eps^{3/2}\sqrt{\delta} H(T,x,y,z) - \eps\delta G(T,x,y,z),
\end{align*}
which reads as
\begin{align}
0& \geq \eps (-\half \theta D_1 \vz) + \eps^{3/2} \half \rho_1\theta_1 D_1^2 \vz - T(\eps C_A D_1\vz + \delta C_B D_1 \vz + \sqrt{\eps\delta}C_C D_1 \vz) \\
& \quad - \eps^2 \left(-\half (\theta_4 - \overline \lambda^2 \theta_5)(\half D_2 + D_1) D_1 \vz + \half \rho_1^2 \theta_3 D_1^3 \vz - T \theta (\half D_2 + D_1 )C_A D_1 \vz\right)   \\
& \quad -\eps^{3/2}\sqrt{\delta}\left(\rho_2 g \theta_2 D_1 \voz_z - \half \theta_6 D_1 \vz - T \theta (\half D_2 + D_1) C_C D_1 \vz\right) \\
& \quad -\eps\delta\left(\rho_2 g \theta_2 D_1 \vzo_z - T \theta (\half D_2 + D_1) C_B D_1 \vz\right). \label{eq:VminusU}
\end{align}
Since all  $\theta_i$ and $g$ are bounded functions and with Proposition~\ref{prop_estimate}, we can again first choose $\eps < \eps''$ and $\delta <  \delta''$ such that the coefficients of $C_i$ are negative, then for $C_i > C_i''$, $i = A, B, C$ the above inequality hold.

Finally, combining the two upper bounded for $(\eps, \delta)$ and for $C_i$, we conclude that the requirements \ref{r1}--\ref{r2} are fulfilled for $\eps \leq \eps' \wedge \eps''$, $\delta \leq \delta' \wedge \delta''$ and for $C_i \geq C_i' \vee C_i''$.

\begin{rem}
 Observing that, under our choices of $N_i$ and $(F, G, H)$,  terms with non-zero terminal values in $V^-$  are either  $\MCO(\eps)$ or $o(\eps)$, one could have chosen $(2T-t)$ for $N_A$, and $(T-t)$ for $(N_B, N_C)$ in the definition \eqref{def:Vpm} of $V^-$. This would eliminate the $C_B$ and $C_C$ terms in \eqref{eq:VminusU}, while the conclusion still holds.
\end{rem}

\subsection{Super-solution}
Similar to the derivation in Section~\ref{sec:sub}, we shall first take care of \ref{r4}. This will allow us to derive the form of $N_i$ in $V^+$ which are also given by 
\begin{equation}\label{def:Nisuper}
N_i = C_i D_1\vz, \quad i = A, B, C.
\end{equation}
With such forms, the requirement \ref{r5} is shown as a consequence of Definition~\ref{def:admissible}, and \ref{r3} will be satisfied with sufficient large choices of $C_i$ and sufficiently small $(\eps, \delta)$.

\subsubsection{The existence and non-positivity of  $\widehat Q[V^+]$  \ref{r4} }\label{sec_Qhat_multi}
Recall the definition of $Q^\pi$ in \eqref{def_Q}, the first order condition gives an optimizer of $Q^\pi[V^+]$, which we denote by $\pi^\ast$:
\begin{equation}\label{eq_piast}
\pi^\ast = - \frac{\lambda(x,y) V_x^+ + \frac{1}{\sqrt{\eps}}\rho_1 a(y) V_{xy}^+ + \sqrt{\delta}\rho_2 g(z) V_{xz}^+}{\sigma(x,y) V^+_{xx}}.
\end{equation}
The requirement \ref{r4} is equivalent to, for all $(t, x,y,z)$ and sufficiently small $(\eps, \delta)$ that:
\begin{enumerate}[ label=(\textbf{R4-\arabic*})]
	\item\label{r4-1} $Q^{\pi^\ast}[V^+] \leq 0$; 
	\item\label{r4-2} $V_{xx}^+ < 0$, so that $\pi^\ast$ is a maximizer and $\widehat Q[V^+] := \sup_\pi Q^\pi[V^+] = Q^{\pi^\ast}[V^+]$.
\end{enumerate}
To this end, let $\pi = \pi^\ast$ in the operator  $Q^\pi$ and apply it to $V^+$ (cf. \eqref{eq:Vpm}):
\begin{align}
&Q^{\pi^\ast}[V^+] \\
& = V_t^+ + \MCL_y (\wtz + \sqrt{\eps}\wthz + \sqrt{\delta} \wto + \eps F + \delta G + \sqrt{\eps\delta} H)\\
& \quad + \sqrt{\delta}\MCM_{yz}\left(\sqrt\eps \wtz + \eps \wthz + \sqrt{\eps\delta}\wto + \eps^{3/2}F + \eps \sqrt{\delta}H + \sqrt{\eps}\delta G\right) + \delta \MCL_z V^+ \nonumber \\
& \quad - \frac{1}{2\vz_{xx}}\left[\lambda \vz_x + \sqrt{\eps}(\lambda \voz_x + \rho_1 a \wtz_{xy}) + \sqrt{\delta}(\lambda\vzo_x + \rho_2 g\vz_{xz})\right. \label{eq_Vx_multi} \\
&\hspace{30pt}+ \delta (\lambda (2T-t)(N_B)_x + \rho_2 g \vzo_{xz}) +  \eps(\lambda (\wtz + (2T-t)N_A)_x + \rho_1 a \wthz_{xy})  \label{eq_Vx_multi2}\\
& \hspace{30pt}+ \left.\sqrt{\eps\delta}(\lambda (2T-t)(N_C)_x + \rho_1 a \wto_{xy} + \rho_2 g \voz_{xz}) + h.o.t.\right]^2 \label{eq_Vx_multi3}\\
& \quad \times \bigg[1 - \sqrt{\eps}\frac{\voz_{xx}}{\vz_{xx}}- \sqrt{\delta}\frac{\vzo_{xx}}{\vz_{xx}}- \eps\Big(\frac{(\wtz + (2T-t)N_A)_{xx}}{\vz_{xx}} - (\frac{\voz_{xx}}{\vz_{xx}})^2\Big) \label{eq_Vxx_multi}\\
& \hspace{30pt} -\delta\Big(\frac{(2T-t)(N_B)_{xx}}{\vz_{xx}} - (\frac{\vzo_{xx}}{\vz_{xx}})^2\Big) - \sqrt{\eps\delta}\Big(\frac{(2T-t)(N_C)_{xx}}{\vz_{xx}} - \frac{2\voz_{xx}\vzo_{xx}}{(\vz_{xx})^2}\Big)\label{eq_Vxx_multi2}\\
& \hspace{30pt}  \eps^{3/2} \MCR^1(t,x,y,z) + \eps\sqrt{\delta} \MCR^2(t,x,y,z) + \sqrt{\eps}\delta \MCR^3(t,x,y,z) + \delta^{3/2} \MCR^4(t,x,y,z)\rule{0cm}{0.7cm}\bigg] \label{eq_Vxx_multi3}\\
& = \eps \text{I}_\eps + \delta \text{I}_\delta + \sqrt{\eps\delta}\text{I}_{\eps\delta} + h.o.t. , \label{def:Vplushot}
\end{align}
where $(\text{I}_\eps, \text{I}_\delta, \text{I}_{\eps\delta})$ are given by
\begin{align}
\text{I}_\eps &= \MCL_y F + \Ltx(\lambda)(\wtz + (2T-t)N_A) - \frac{1}{2\vz_{xx}}\Big(\lambda \vz_x\frac{\voz_{xx}}{\vz_{xx}} - \lambda \voz_x - \rho_1 a \wtz_{xy}\Big)^2 \\
& \quad + \rho_1 a\lambda D_1 \wthz_y, \\
\text{I}_\delta &= \MCL_y G + \Ltx(\lambda)((2T-t)N_B) - \frac{1}{2\vz_{xx}}\Big(\lambda \vz_x\frac{\vzo_{xx}}{\vz_{xx}} - \lambda \vzo_x - \rho_2 g \vz_{xz}\Big)^2 \\
& \quad + \rho_2 g\lambda D_1 \vzo_z + \MCL_z \vz, \\
\text{I}_{\eps\delta} & =  \MCL_y H + \Ltx(\lambda)((2T-t)N_C) + \MCM_{yz} \wtz + \rho_1 a\lambda D_1 \wto_y + \rho_2 g\lambda D_1 \voz_z  \\
& \quad -\frac{1}{\vz_{xx}}\Big(\lambda \vz_x\frac{\voz_{xx}}{\vz_{xx}} - \lambda \voz_x - \rho_1 a \wtz_{xy}\Big)\Big(\lambda \vz_x\frac{\vzo_{xx}}{\vz_{xx}} - \lambda \vzo_x - \rho_2 g \vz_{xz}\Big).
\end{align}
As in the sub-solution case, we first show that with the choice \eqref{def:Nisuper}, $\average{\text{I}_\eps}, \average{\text{I}_\delta}$ and $\average{\text{I}_{\eps\delta}}$ can be strictly negative. Then by letting $(F, G, H)$ be the solution of $(\text{I}_\eps, \text{I}_\delta, \text{I}_{\eps\delta}) - (\average{\text{I}_\eps}, \average{\text{I}_\delta}, \average{\text{I}_{\eps\delta}}) = 0$, we have terms at $\MCO(\eps + \delta)$ are strictly negative. The computation and reasoning are very similar to the sub-solution case, thus we omit here and summarize the results. 

Regarding $\MCO(\eps)$ terms, one has
\begin{align}
\average{\text{I}_\eps}  & = -N_A  +  \average{\Ltx(\lambda)\wtz}  + \average{\rho_1 a \lambda D_1\wth_y} \\
& \quad  - \frac{1}{2\vz_{xx}}\bigg(\overline \lambda^2 \Big(\vz_x\frac{\voz_{xx}}{\vz_{xx}}-\voz_x\Big)^2 + \rho_1 B(z) \Big(\vz_x\frac{\voz_{xx}}{\vz_{xx}}-\voz_x\Big) (D_1\vz)_x \\
& \quad + \frac{1}{4}\rho^2_1 \average{a^2\theta'^2}[(D_1\vz)_x]^2\bigg). \\
F(t, &x, y, z)  \\
&= -\half (T-t) \rho_1^2 \theta_1(y,z) B(z)D_1(\half D_2 + D_1) D_1^2\vz - \half \rho_1^2 \theta_3(y, z) D_1^3\vz \\
&\quad  + \half (\theta_4(y,z)-\overline \lambda^2\theta_5(y,z))(\half D_2 + D_1) D_1\vz - (2T-t) \theta(y, z) (\half D_2 + D_1)C_AD_1\vz \nonumber\\
& \quad + \frac{1}{2\vz_{xx}}\bigg(\theta(y, z) \Big(\vz_x\frac{\voz_{xx}}{\vz_{xx}}-\voz_x\Big)^2 + \rho_1 \theta_1(y, z) \Big(\vz_x\frac{\voz_{xx}}{\vz_{xx}}-\voz_x\Big) (D_1\vz)_x  \\
& \hspace{40pt}  + \frac{1}{4}\rho_1^2 \theta_9(y, z) [(D_1\vz)_x]^2\bigg),
\end{align}
where $\theta_9(y, z)$ solves the Poisson equation \eqref{def:theta9}.

For terms of  $\MCO(\delta)$, we have
\begin{align}
\average{\text{I}_\delta} &= \MCL_z \vz + \rho_2 g(z)\widehat \lambda D_1 \vzo_z - N_B  \\
& \quad - \frac{1}{2\vz_{xx}}\bigg(\overline \lambda^2 \Big(\vz_x\frac{\vzo_{xx}}{\vz_{xx}}-\vzo_x\Big)^2 - 2 \rho_2 g \widehat \lambda \Big(\vz_x\frac{\vzo_{xx}}{\vz_{xx}}-\vzo_x\Big) \vz_{xz}  + \rho^2_2g^2 [\vz_{xz}]^2\bigg).\\
G(t,&x,y,z) \\
&= -\rho_2 g\theta_2(y,z) D_1 \vzo_z - (2T-t) \theta(y,z) (\half D_2 + D_1) C_B D_1 \vz \\
& \quad + \frac{1}{2\vz_{xx}}\bigg(\theta(y, z) \Big(\vz_x\frac{\vzo_{xx}}{\vz_{xx}}-\vzo_x\Big)^2 - 2\rho_2 g \theta_2(y,z) \Big(\vz_x\frac{\vzo_{xx}}{\vz_{xx}}-\vzo_x\Big) \vz_{xz}\bigg).
\end{align}

Regarding $\MCO(\sqrt{\eps\delta})$ terms, we deduce
\begin{align}
\average{\text{I}_{\eps\delta}} &= \average{\MCM_{yz}\wtz} + \rho_1 \average{a\lambda D_1 \wto_y} + \rho_2 g(z)\widehat \lambda D_1 \voz_z - N_C  \\
& \quad - \frac{1}{\vz_{xx}}\bigg(\overline \lambda^2 \Big(\vz_x\frac{\voz_{xx}}{\vz_{xx}}-\voz_x\Big)\Big(\vz_x\frac{\vzo_{xx}}{\vz_{xx}}-\vzo_x\Big)  - \rho_2 g \widehat \lambda \Big(\vz_x\frac{\voz_{xx}}{\vz_{xx}}-\voz_x\Big) \vz_{xz}\bigg) \\
& \quad - \frac{1}{\vz_{xx}}\bigg(\half \rho_1 B(z)  \Big(\vz_x\frac{\vzo_{xx}}{\vz_{xx}}-\vzo_x\Big) (D_1\vz)_x - \half \rho_1\rho_2 g \average{\theta_y} \vz_{xz}(D_1\vz)_x\bigg).
\end{align}
with the first two terms calculated as
\begin{align}
\average{\MCM_{yz}\wtz} & = -\half \rho_1 g \left(\average{a\theta_{yz}} D_1 \vz + \average{a\theta_y}(D_1\vz)_z\right), \\
\average{a\lambda D_1 \wto_y} & = -\half (T-t)^2 \rho_2 B(z) \widehat \lambda \overline \lambda \overline \lambda' g(z) D_1 (\half D_2 + D_1) D_1^2 \vz\\
& \hspace{100pt} - \rho_2 \average{a\lambda \theta_{2y}} g(z) (T-t) \overline \lambda \overline \lambda' D_1^3 \vz.
\end{align}
Then the function $H(t, x, y, z)$ is identified as:
\begin{align}
H &= -\half \rho_1 g \theta_{10}(y,z) D_1 \vz +\half \rho_1 g \theta_{11}(y,z)(D_1\vz)_z \\
&\quad + \rho_1 \half (T-t)^2 \rho_2 \theta_1(y,z) \widehat \lambda \overline \lambda \overline \lambda' g(z) D_1 (\half D_2 + D_1) D_1^2 \vz   - \rho_2 g(z) \theta_2(y,z) D_1 \voz_z\\
& \quad + \rho_1\rho_2 \theta_8(y,z) g(z) (T-t)\overline \lambda \overline{\lambda}' D_1^3\vz  - (2T-t) \theta(y,z)(\half D_2 + D_1) C_C D_1 \vz \\
& \quad + \frac{1}{\vz_{xx}}\bigg(\theta(y,z)\Big(\vz_x\frac{\voz_{xx}}{\vz_{xx}}-\voz_x\Big)\Big(\vz_x\frac{\vzo_{xx}}{\vz_{xx}}-\vzo_x\Big)  \\
& \quad \quad \quad \quad \quad \quad - \rho_2 g \theta_2(y,z) \Big(\vz_x\frac{\voz_{xx}}{\vz_{xx}}-\voz_x\Big) \vz_{xz}\bigg) \\
& \quad + \frac{1}{\vz_{xx}}\bigg(\half \rho_1 \theta_1(y,z)  \Big(\vz_x\frac{\vzo_{xx}}{\vz_{xx}}-\vzo_x\Big) (D_1\vz)_x - \half \rho_1\rho_2 g \theta_7(y,z) \vz_{xz}(D_1\vz)_x\bigg).
\end{align}
where $\theta_{10}$ and $\theta_{11}$ solve the Poisson equations \eqref{def:theta10} and \eqref{def:theta11} respectively.

Next, we show that all high order terms in \eqref{def:Vplushot} can be eliminated by the strictly negative terms $\eps \text{I}_\eps + \delta \text{I}_\delta + \sqrt{\eps\delta} \text{I}_{\eps\delta}$ at order $\eps + \delta$ by increasing $N_i$, thus \ref{r4-1} is fulfilled. To proceed further, we need the following lemmas, which are obtained by lengthy but straightforward calculations, and thus omitted. 
\begin{lem}\label{lem_diff_multi}
	Under standing assumptions, we have the following estimates:
	\begin{equation}
		(\MCD \vz)_{xx} \leq f(y,z) \vz_{xx}, \quad   (\MCD \vz)_x \leq f(y,z) \vz_x,
	\end{equation}	
	where $f(y, z)$ is a bounded function, and $\MCD$ takes the following operators:
	\begin{equation}
	D_1, D_1^2, D_1^3, D_1^4, D_2D_1, D_2D_1^2, D_1D_2D_1^2.
	\end{equation}
\end{lem}

\begin{lem}
	Under standing assumptions, we have the following estimates (with bounded $f(y,z)$):
	\begin{align}
	&	(D_1 \vzo_z)_{xx} \leq f(y,z) \vz_{xx}, \quad (D_1\vzo_z)_{x} \leq f(y,z) \vz_x,\\
	&	(D_1 \voz_z)_{xx} \leq f(y,z) \vz_{xx}, \quad (D_1\voz_z)_{x} \leq f(y,z) \vz_x, \\
	&	(D_1 \vz)_{zxx} \leq f(y,z) \vz_{xx}, \quad (D_1\vz)_{x} \leq f(y,z) \vz_x
	\end{align}	
\end{lem}

\begin{lem}\label{lem_diffFGH_multi}
	Under standing assumptions, we have the following estimates (with bounded $f(y,z)$):
	\begin{align}
	& F_{x} \leq f(y,z) \vz_{x}, G_{x} \leq f(y,z) \vz_{x} \text{ and } H_{x} \leq f(y,z) \vz_{x}, \\
	&	F_{xx} \leq f(y,z) \vz_{xx}, G_{xx} \leq f(y,z) \vz_{xx} \text{ and } H_{xx} \leq f(y,z) \vz_{xx}.
	\end{align}	
\end{lem}

Following Lemmas~\ref{lem_diff_multi}--\ref{lem_diffFGH_multi} and Proposition~\ref{prop_estimate}--\ref{prop_risktolerance}, the terms in \eqref{eq_Vx_multi}--\eqref{eq_Vx_multi3} are bounded by a multiple of $\vz_x$, and terms in \eqref{eq_Vxx_multi}--\eqref{eq_Vxx_multi2} are bounded by constants, both depending linearly in $C_i$. Terms in \eqref{eq_Vxx_multi3} 
are bounded in $(t, x, y, z)$ for any $\eps < \eps'$, $\delta < \delta'$ and $C_i > C_i'$, due to the boundedness of $v^{(i,j)}_{xx}/ \vz_{xx}$, $w^{(i,j)}/\vz_{xx}$, etc, and the asymptotically behavior $\MCO(1/C_i)$ as $C_i \to \infty$, $i = A, B, C$. Other terms in \eqref{def:Vplushot} involving $\partial_t$, $\MCM_{yz}$ and $\MCL_z$ can be are verified by direct differentiations. To summarize, all terms higher than $\MCO(\eps + \delta)$ in $ Q^{\pi^\ast}[V^+]$ are bounded by functions of the form $\eps^\alpha\delta^\beta f(y,z) D_1\vz$, where $\alpha + \beta \geq 3/2$, and $f(y,z)$ is bounded in $(y,z)$ and at most linearly growth in $C_i$. Thus, one can first choose small $\eps < \eps''$ and $\delta < \delta''$ so that the coefficient in front of $C_i$ are negative, then by letting $C_i > C_i''$, we will have $Q^{\pi^\ast}[V^+] \leq 0$.

Lastly, we show \ref{r4-2}, that is, $V_{xx}^+ < 0$. Observing that 
\begin{align}
V_{xx}^+ &= \vz_{xx}\bigg(1 + \sqrt\eps \frac{\voz_{xx}}{\vz_{xx}} + \sqrt{\delta} \frac{\vzo_{xx}}{\vz_{xx}} + \eps \frac{\wtz_{xx}}{\vz_{xx}} + \eps^{3/2}\frac{\wthz_{xx}}{\vz_{xx}} + \eps\sqrt{\delta}\frac{\wto_{xx}}{\vz_{xx}} + \eps^2 \frac{F_{xx}}{\vz_{xx}} + \eps\delta \frac{G_{xx}}{\vz_{xx}}  \\
&\hspace{30pt} + \eps^{3/2}\sqrt{\delta} \frac{H_{xx}}{\vz_{xx}} + (2T-t)\Big(\frac{\eps (N_A)_{xx} + \delta (N_B)_{xx} + \sqrt{\eps\delta}(N_C)_{xx} }{\vz_{xx}}\Big)\bigg), \label{eq_Vxx}
\end{align}
and we recall that $\vz_{xx} <0$ by the concavity of  classic Merton problem. As a consequence of Lemmas~\ref{lem_diff_multi}--\ref{lem_diffFGH_multi}, all ratios in \eqref{eq_Vxx} are bounded in $(y, z)$. Therefore, for given $C_i$, $i = A, B, C$, one can choose sufficiently small $\eps < \eps'''$ and $\delta < \delta'''$ such that the sum in the paraphrases stays positive, and consequently $V^+_{xx} < 0$ for all $(t,x,y,z)$. 

Now by first taking $C_i = \max\{C_i', C_i''\} + 1$, then determine $\eps'''$ and $\delta'''$, and finally taking $\eps < \min\{\eps', \eps'', \eps'''\}$, $\delta < \min\{\delta', \delta'', \delta'''\}$, \ref{r4} is fulfilled.

\subsubsection{The requirement \ref{r3}: $U(x) \leq V^+(T, x, y, z)$}
The argument here is parallel to Section~\ref{sec:subr1}. So one can first choose $\eps < \eps''''$ and $\delta < \delta''''$ and then $C_i > C_i'''$ so that \ref{r3} holds.

Therefore, to let both $\ref{r3}$ and $\ref{r4}$ hold, we need first take $C_i = \max\{C_i', C_i'', C_i'''\} + 1$, then determine $\eps'''$ and $\delta'''$ so that \eqref{eq_Vxx} is negative, and finally take $\eps < \min\{\eps', \eps'', \eps''', \eps''''\}$, $\delta < \min\{\delta', \delta'', \delta''', \delta''''\}$.

\subsubsection{The martingality of It\^{o} integral terms \ref{r5}}

Following Lemma~\ref{lem_diff_multi}--\ref{lem_diffFGH_multi}, all functions $v^{(i,j)}_x$, $w^{(i,j)}_x$, $(N_A, N_B, N_C)_x$ and $(F, G, H)_x$ are bounded by $f(y, z)\vz_x$, where $f(y,z)$ is a bounded function in $(y, z)$. Therefore, for a given $\pi$, the first It\^{o} integral is a true martingale if
\begin{equation}
\EE\int_0^T \left(\pi(t, X_t^\pi, Y_t, Z_t)\sigma(Y_t, Z_t)\vz_x(t, X_t^\pi, Y_t, Z_t)\right)^2 \ud t < \infty,
\end{equation}
which is automatically satisfied by any admissible control $\pi$ by \eqref{admissibility1} (cf. Definition~\ref{def:admissible}).

For the rest two to be true martingales, we need
\begin{equation}
\EE\int_0^T \left(D_1\vz(t, X_t^\pi, Y_t, Z_t)\right)^2 \ud t < \infty,
\end{equation}
which is part of the definition of admissibility \eqref{admissibility2}. Therefore, we obtain the desired result.

\section{Conclusion}\label{sec:conclusion}
This paper provides the accuracy analysis of asymptotics for the portfolio optimization problem with general utility functions and two (fast and slow) stochastic factors. This sets up the theoretical foundation of using asymptotic expansion to derive approximations for value functions and optimal strategies in the regime where these factors are running on both slow and fast timescales. Specifically, we construct the sub- and super-solutions to the fully nonlinear problem so that their difference is at the desired level of accuracy. In the present context, the fast varying factor requires a careful treatment of the singular perturbation for a fully nonlinear equation. Moreover, the proofs presented here can be adapted to justify other derivations of accuracy in various contexts as in \cite{fouque2021optimal} for instance.

\section*{Acknowledgment} 

The authors are grateful to Professor Thaleia Zariphopoulou for useful discussions on the admissibility of $\pz$ and to Professor Maxim Bichuch. JPF was supported by NSF grant  DMS-1814091. RH was partially supported by the NSF grant DMS-1953035, and the Faculty Career Development Award, the Research Assistance Program Award, and the Early Career Faculty Acceleration funding at UCSB.

\bibliographystyle{plain}
\bibliography{Reference}

\end{document}